\documentclass[aps,superscriptaddress,twocolumn,twoside,floatfix,prl,a4paper]{revtex4-2} 
\usepackage{times}
\usepackage{epsfig}
\usepackage{amsfonts}
\usepackage{amsmath}
\usepackage{amssymb,amsthm}
\usepackage{xcolor,colortbl}
\usepackage{multirow}
\usepackage{braket}
\usepackage{latexsym}
\usepackage{amsfonts}
\usepackage{mathrsfs}
\usepackage{natbib}
\usepackage{verbatim}
\usepackage{gensymb}
\usepackage{caption}
\usepackage{caption}
\usepackage{subcaption}
\usepackage{ragged2e}
\DeclareCaptionJustification{justified}{\justifying}
\usepackage{blkarray}
\usepackage{graphicx}
\usepackage[colorlinks=true,linkcolor=blue,citecolor=magenta,urlcolor=blue]{hyperref}
\allowdisplaybreaks
\newtheorem{theorem}{Theorem}

\newtheorem{proposition}{Proposition}

\begin{document}

\title{Distilling Nonlocality in Quantum Correlations}

\author{Sahil Gopalkrishna Naik}
\affiliation{Department of Physics of Complex Systems, S.N. Bose National Center for Basic Sciences, Block JD, Sector III, Salt Lake, Kolkata 700106, India.}

\author{Govind Lal Sidhardh}
\affiliation{Department of Physics of Complex Systems, S.N. Bose National Center for Basic Sciences, Block JD, Sector III, Salt Lake, Kolkata 700106, India.}

\author{Samrat Sen}
\affiliation{Department of Physics of Complex Systems, S.N. Bose National Center for Basic Sciences, Block JD, Sector III, Salt Lake, Kolkata 700106, India.}

\author{Arup Roy}
\affiliation{Department of Physics, A B N Seal College Cooch Behar, West Bengal 736101, India}

\author{Ashutosh Rai}
\affiliation{School of Electrical Engineering, Korea Advanced Institute of Science and Technology, Daejeon 34141, Republic of Korea}
\email{ashutosh.rai@kaist.ac.kr}

\author{Manik Banik }
\affiliation{Department of Physics of Complex Systems, S.N. Bose National Center for Basic Sciences, Block JD, Sector III, Salt Lake, Kolkata 700106, India.} 
\email{manik11ju@gmail.com}

\begin{abstract}
Nonlocality, as established by seminal Bell's theorem, is considered to be the most striking feature of correlations present in space like separated events. Its practical application in device independent protocols, such as secure key distribution, randomness certification, {\it etc.}, demands identification and amplification of such correlations observed in the quantum world. In this Letter we study the prospect of nonlocality distillation, wherein, by applying a natural set of free operations (called wirings) on many copies of weakly nonlocal systems, one aims to generate  correlations of higher nonlocal strength. In the simplest Bell scenario, we identify a protocol, namely, logical OR-AND wiring, that can distil nonlocality to significantly high degree starting from arbitrarily weak quantum nonlocal correlations. As it turns out, our protocol has several interesting facets: (i) it demonstrates that set of distillable quantum  correlations has non zero measure in the full eight-dimensional correlation space, (ii) it can distil quantum Hardy correlations by preserving its structure, (iii) it shows that (nonlocal) quantum correlations sufficiently close to the local deterministic points can be distilled by a significant amount. Finally, we also demonstrate efficacy of the considered distillation protocol in detecting postquantum correlations.
\end{abstract}

\maketitle

{\it Introduction.--} One of the most celebrated non-classical aspects of quantum mechanics was pioneered by J. S. Bell in the year 1964 \cite{Bell1964} (see also \cite{Bell1966}). Bell's theorem mandates departure of quantum theory from the {\it locally causal} world view which subsequently has been confirmed in several milestone experiments led by Clauser, Aspect, Zeilinger, and others \cite{Clauser1969,Freedman1972,Aspect1976,Aspect1981,Aspect1982,Aspect1982(1),Zukowski93,Pan1998,Weihs1998,Bouwmeester1999}. Unlike other non-classical features, such as entanglement and coherence, study of nonlocality can be conducted in a device independent setting where only the input-output statistics of the device matters  and one does not need to know the inner design or working mechanisms of the device \cite{Scaran2012}. Along with foundational implications, Bell nonlocality has also been identified as the necessary resource for several important protocols \cite{Ekert1991,Barrett2009,Pironio2009,Pironio2010,Colbeck2011,Colbeck2012,Chaturvedi2015,Mukherjee2015,Pappa2015,Roy2016,Frenkel2022,Patra2022}, which, thus, makes the question of refinement or distillation of this resource practically indispensable. Study of nonlocality distillation has two major implications -- (i) practical: where one aims to distil nonlocal correlations observed in the quantum world which can be then applied to make information flow networks efficient and secure, and (ii) foundational: where the goal is to identify post quantum correlations, which, in turn, helps to understand the speciality of quantum theory among other possibilities allowed within the framework of generalized probabilistic theories. Interestingly, in Ref.\cite{Forster2009}, Forster {\it et al.} proposed a nonlocality distillation protocol that can extract nonlocality in stronger form starting with many copies of weakly nonlocal systems; this work has inspired a number of subsequent works consisting of interesting results on nonlocality distillation \cite{Brunner2009,Short2009,Brunner2011,Forster2011,HoyerJibran2010,Jibran2012,HoyerJibran2013,Ebbe2014,Tuziemski2015,Beigi2015,Brassard2015,Brito2019,Eftaxias2022,Allcock2009,Lang2014}.

The research conducted so far on nonlocality distillation is mainly focused on distilling post quantum correlations \cite{Brunner2009,Short2009,Brunner2011,Forster2011,HoyerJibran2010,Jibran2012,Ebbe2014,Tuziemski2015,Beigi2015,Brassard2015,Brito2019,Eftaxias2022}. Only a few protocols are known that successfully distil some quantum correlations \cite{Forster2009,Eftaxias2022}. The difficulty arises due to the top-down approach considered in earlier works where one starts with some parametric family of generic no-signaling (NS) correlations, and after obtaining a successful distillation protocol the aim is to check whether for some range of the parameter values the considered NS correlations allow quantum realization or not. For the simplest bipartite case, the well known analytical criterion by Tsirelson-Landau-Masanes \cite{Tsirelson1987,Landau1988,Masanes2008} and the Navascues-Pironio-Acin (NPA) criterion \cite{Navascues2007}, and in general case a hierarchy of semi-definite programming conditions \cite{Navascues2008} can serve this purpose. Only in some fortunate cases  sophisticated choices of the parametric class of NS correlations might lead to a desirable subset of quantum realizable correlations. However, the approach has severe pitfall when more input-output scenarios are considered, as the recent mathematical breakthrough by W. Slofstra and the subsequent results establish that the set of quantum correlations is not topologically closed \cite{Slofstra2020,Slofstra2019,Dykema2019}. There are only a few results that report distillation of nonlocal correlations within quantum setup \cite{Forster2009,Eftaxias2022}, albeit the nonlocal strength of the distilled correlation is low. Therefore the aspects of analytical and quantitative study for distillation of quantum nonlocal correlations remain open.

In this letter, we propose a generic framework for nonlocality distillation which overcomes limitations of the thus far proposed protocols. In contrast to the previously reported results, we intend to find out efficient distillation protocol(s) for quantum correlations. To this aim we consider the bottom-up approach. Instead of generic NS signaling correlations we start with weak  nonlocal correlations which is quantum, and then try to obtain a nonlocality distillation protocol. The set of quantum correlations being closed under {\it wirings} \cite{Allcock2009,Lang2014} assures the resulting distilled correlations to be quantum. Interestingly, we identify a simple protocol and come up with a generic approach that successfully distil nonlocality in a large class of weakly nonlocal quantum correlations. Towards this goal, first we consider a variant of nonlocality test proposed by Lucien Hardy \cite{Hardy1993}. Success probability in Hardy's test qualifies as a measure of nonlocality for Hardy's correlations~\cite{Cereceda2000}. Given two copies of a quantum Hardy correlation, we show that there exist a simple wiring that can distil Hardy nonlocality. We call this wiring logical OR-AND protocol, where OR~($\vee$) and AND~($\wedge$) functions on $2$-bits $z_1,z_2$ are defined as $\vee(z_1,z_2)=\max\{z_1,z_2\}$ and $\wedge(z_1,z_2)=\min\{z_1,z_2\}$, respectively. The OR-AND protocol allows an immediate $n$-copy generalization (see Fig.\ref{Fig:OR-AND}), which can provide a substantial distillation of Hardy's success with a sufficiently large copies of initial correlations. Further, we show that the OR-AND wiring when applied to a broader class of quantum correlations yields an interesting result: an arbitrarily small violation of the Clauser-Horne-Shimony-Holt (CHSH) \cite{Clauser1969} inequality can be amplified to a significantly higher degree. Finally, by applying our protocol we demonstrate that nonlocal correlations arbitrarily close to the extreme points of the set of local correlations are always distilled, which, in turn, establishes that set of distillable quantum correlations has non-zero measure in the full eight dimensions of the correlation space. We also study distillation of post quantum correlations, and show that OR-AND protocol becomes efficient there too. In particular, we find correlations whose post-quantum signature is established through OR-AND distillation, while the known information principles, such as nontrivial communication complexity \cite{Brassard2005} and information causality \cite{Pawlowski2009,Miklin2021}, fail to serve the purpose. In the end, we discuss novelty of the approach followed here in comparison to the existing methods on nonlocality distillation.
\begin{figure}[t!]
\centering
\includegraphics[width=0.4\textwidth]{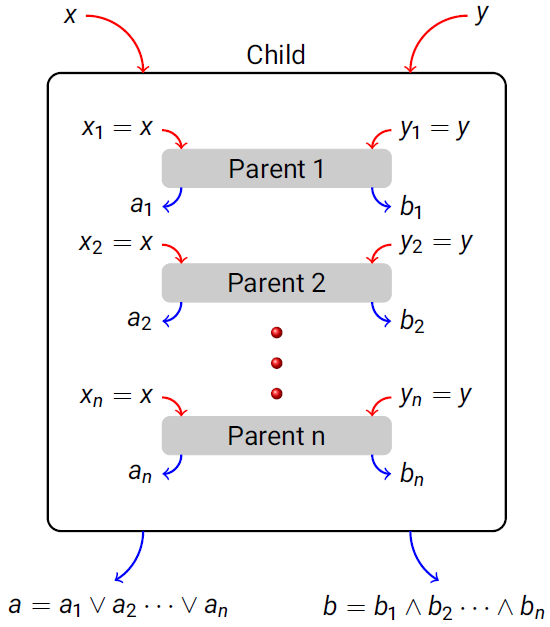}
\caption{Multi-copy OR-AND wiring. Given $n$-number of parent correlations $\{\mathrm{P}_{NS}[i]\}_{i=1}^n\subset\mathcal{NS}$, the OR-AND wiring produce a child correlation $\mathrm{P}^{(n)}_{NS}\in\mathcal{NS}$. The outcome $a$ on Alice's side for the child box is obtained as, $a=a_1\vee\cdots\vee a_n=\max\{a_1,\cdots,a_n\}$ for the input $x_1=\cdots=x_n=x$, where $x_i$ and $a_i$ are the input and output of the $i^{th}$ parent. On the Bob's side, $y_1=\cdots=y_n=y$ and $b=b_1\wedge \cdots\wedge b_n=\min\{b_1,\cdots,b_n\}$.}
\label{Fig:OR-AND}
\end{figure}

{\it Preliminaries.--} A Bell experiment involves spatially separated parties performing local measurements on the respective parts of a composite system shared among them. The simplest case considers two parties, Alice and Bob, with their respective inputs to the box denoted as $x$ and $y$, and outputs from the box denoted as $a$ and $b$, respectively, where $x,y,a,b\in\{0,1\}$; and this is generally called the 2-2-2 Bell scenario. Correlation generated by a box $\mathrm{P}$ is the set of joint input-output probabilities, {\it i.e.}, $\mathrm{P}\equiv\{p(ab|xy)\}$. Set of boxes satisfying the no-signaling (NS) condition forms an $8$-dimensional polytope $\mathcal{NS}$ having $24$ vertices~\cite{Barrett2005}: $8$ non-local vertices (Popescu-Rohrlich (PR) boxes \cite{Popescu1994}) given by $\mathrm{P}^{\alpha\beta\gamma}_{NL}\equiv\{p(ab|xy)=\frac{1}{2}\delta_{(a\oplus b,~xy\oplus \alpha x\oplus \beta y \oplus \gamma)}\}$ where $\alpha,\beta,\gamma \in\{0,1\}$, and $16$ local deterministic vertices given by $\mathrm{P}^{\alpha_1\alpha_2\beta_1\beta_2}_{L}\equiv\{p(ab|xy)=\delta_{(a,\alpha_1 x \oplus \alpha_2)}~\delta_{(b,\beta_1 y \oplus \beta_2)}\}$ where $\alpha_1,\alpha_2,\beta_1,\beta_2 \in\{0,1\}$. A correlation is termed as local {\it if and only if} it allows a decomposition of the form $p(ab|xy)=\sum p(\lambda)p(a|x,\lambda)p(b|y,\lambda)$, where $\lambda$ is some classical variable shared between Alice and Bob following a distribution $\{p(\lambda)\}$ \cite{Bell1964} (see also \cite{Bell1966}). Collection of all local correlations forms a sub polytope $\mathcal{L}$, within $\mathcal{NS}$, with $16$ local deterministic boxes as their vertices. Correlations obtained from local quantum measurements performed on some bipartite quantum state are called quantum correlations. Set of all quantum correlations $\mathcal{Q}$ forms a convex set (but not a polytope) lying strictly in between local and NS polytope, {\it i.e.}, $\mathcal{L}\subsetneq\mathcal{Q}\subsetneq\mathcal{NS}$ \cite{Goh2018}. No signaling correlations that do not belong to the set $\mathcal{L}$ are called nonlocal as they do not allow a local-causal description \cite{Bell1964}. We consider one of the eight symmetries for the nonlocal correlations witnessed by the Bell CHSH inequality
\begin{align}
\mathcal{B}\equiv \langle00\rangle-\langle01\rangle-\langle10\rangle-\langle11\rangle \leq 2,
\end{align}
where $\langle xy\rangle:=\sum_{a,b}(-1)^{a\oplus b}~p(ab|xy)$. Then, only one PR-box violates the inequality and there are exactly eight extremal local boxes that saturate the local bound. These nine boxes, which form an eight dimensional simplex~\cite{Rai2019(1)}, are as follows
\begin{align}
\left\{\mathrm{P}^{110}_{NL}\right\}\!\Rightarrow\!\mathcal{B}\!=\!4,~
\left\{\!\begin{aligned}
\mathrm{P}^{0001}_{L_1},~\mathrm{P}^{0100}_{L_2},~\mathrm{P}^{0111}_{L_3},\\
\mathrm{P}^{1101}_{L_4},~\mathrm{P}^{1111}_{L_5},~\mathrm{P}^{1000}_{L_6},\\
\mathrm{P}^{0010}_{L_7},~\mathrm{P}^{1010}_{L_8}~~~~~~~~~~~~~
\end{aligned}\right\}\!\Rightarrow\!\mathcal{B}\!=\!2.\label{symmetry}
\end{align}
The CHSH value $\mathcal{B}<2$ for all the remaining extremal nonlocal or local boxes. The choice of this symmetry is due to the simple OR-AND function description of our distillation protocol. One can always apply local reversible relabellings to switch to any of the eight symmetries (then the protocol become suitably relabelled OR-AND ). From here on, we will simply refer to the nine boxes in Eq.~(\ref{symmetry}) by dropping their superscripts. 

While violation of the Bell-CHSH inequality is a paradigmatic test for certifying nonlocality of a correlation, Hardy provided a simpler nonlocality argument for the same \cite{Hardy1993}. He showed that if the following four conditions
\begin{align}
p_{_{\text{Hardy}}}&\equiv p(0,0|A_0,B_0)>0, \label{h1}\\
p(0,0|A_0,B_1)&=p(0,0|A_1,B_0)=p(1,1|A_1, B_1)=0,\nonumber    
\end{align}
are satisfied, then the resulting correlation is necessarily nonlocal. The probability $p_{_{\text{Hardy}}}$ in Eq.(\ref{h1}) quantifies strength of nonlocality of Hardy correlations. While the maximum value of $p_{_{\text{Hardy}}}$ in no-signaling set is $1/2$ (achieved with PR box), its optimal value in quantum mechanics turns out to be $(5\sqrt{5}-11)/2\approx0.09$ \cite{Rabelo2012} (see also \cite{Seshadreesan2011, Rai2022}), and it is achieved on a pure two qubit state with projective measurements. The correlation yielding the maximum Hardy nonlocality in quantum theory reads
\begin{align}
\mathrm{H}_{Q}^{\max}&=(5\sqrt{5}-11)~\mathrm{P}_{NL}\nonumber\\
&+\frac{1}{2}(7-3\sqrt{5})~\sum_{i=1}^4\mathrm{P}_{L_i}+(\sqrt{5}-2)~\mathrm{P}_{L_5},
\end{align}
and it has been shown to be an extreme point of the set $\mathcal{Q}$ \cite{Goh2018}. 

{\it Distillation of quantum nonlocality.--} In the the following parts of this letter, we first present our results on distillation of quantum nonlocal correlations. For ease of readability, here we present our main results and discuss the important proofs. The detailed analyses of some of the more technical aspects is deferred to the Supplemental. We start by presenting our results on  nonlocality distillation of Hardy's correlations, by applying the OR-AND protocol. 
\begin{theorem}\label{theo1}
The OR-AND wiring preserves the structure of quantum Hardy correlations and can efficiently distil the strength of success probability in Hardy's test of nonlocality.
\end{theorem}
\begin{proof}
Within the considered CHSH symmetry [{\it i.e.}, Eq.(\ref{symmetry})], any NS correlation can be represented as a convex mixture of the nonlocal vertex along with $8$ local vertices. A quantum Hardy correlation demands nonzero weights for exactly $5$ different local vertices along with the nonlocal vertex \cite{Brito2019}. This, in turn, ensures that the correlation matrix has precisely $3$ zero elements corresponding to probabilities taking the value zero in the Hardy's test. It turns out that, then any quantum Hardy correlation  $\mathrm{H}_Q$, can be represented as 
\begin{align}
\mathrm{H}_Q&\!=\!c_0\mathrm{P}_{NL}\!+\!\!\sum_{i=1}^5c_i\mathrm{P}_{L_i};~~c_i>0~\forall~i,~~\mbox{and}~\sum_{i=0}^5 \!c_i\!=\!1,\label{q-hardy}
\end{align}
with success in Hardy's test $p_{_\text{Hardy}}=\frac{c_0}{2}$. On applying the OR-AND wiring to $2$ copies of (parents) $\mathrm{H}_Q$ one obtains a resulting (child) Hardy correlation $\mathrm{H}^{(2)}_Q$ with nonlocality strength $p^{(2)}_{_\text{Hardy}}=(\frac{c_0}{2}+c_1)^2-c_1^2$. Similarly, on applying OR-AND protocol to $n$ copies of $\mathrm{H}_Q$ we obtain a (child) Hardy correlation $\mathrm{H}^{(n)}_Q$ with nonlocality strength $p^{(n)}_{_\text{Hardy}}=(\frac{c_0}{2}+c_1)^n-c_1^n$. The protocol serves the purpose of an effective $n$ copy distillation as long as $p^{(k)}_{_\text{Hardy}}>p^{(k-1)}_{_\text{Hardy}}$, for all $k\in \{2,3,...,n\}$. For example, let us consider the class of quantum correlations
\begin{align}
\tilde{\mathrm{H}}_Q(\lambda):=\lambda\mathrm{H}^{\max}_Q+(1-\lambda)\mathrm{P}_{L_1},~~\lambda\in(0,1],\label{lambda-hardy}
\end{align}
for which the Hardy success probability $p_{_\text{Hardy}}(\lambda)=\lambda k_1=\frac{\lambda}{2}(5\sqrt{5}-11)$. Then, on applying the OR-AND protocol  to two copies of (parent) $\tilde{\mathrm{H}}_Q(\lambda)$ yields the child $\tilde{\mathrm{H}}^{(2)}_Q(\lambda)$ with Hardy success probability $p^{(2)}_{_\text{Hardy}}(\lambda)=[\lambda (k_1+2k_2)+2(1-\lambda)]\lambda k_1$. Then it follows that $p^{(2)}_{_\text{Hardy}}(\lambda)>p_{_\text{Hardy}}(\lambda)$ for $\lambda\in(0,\phi^{-1})$, where $\phi$ is the golden ratio, {\it i.e.}, $\phi=\frac{1+\sqrt{5}}{2}$. On considering a sufficiently large copies ($n$) of very weakly nonlocal correlations $\tilde{\mathrm{H}}_Q$ with $p_{_\text{Hardy}}(\lambda) \rightarrow 0$, we find that the OR-AND wiring results in a Hardy correlation with a considerably large nonlocal strength $p^{(n)}_{_\text{Hardy}}=0.041$ (see Supplemental).   
\end{proof}
In the Supplemental we further shown that arbitrarily weak quantum Hardy correlation can be distilled up to $0.0433$. Since a correlation with Hardy success $p_{_{\text{Hardy}}}$ yields CHSH value $2+4p_{_{\text{Hardy}}}$ \cite{Cereceda2000}, in order to measure the efficacy of distillation let us define the \emph{Tsirelson gain in percentage} as follows,
\begin{align}
\Delta\mathcal{T}:=\frac{1}{2(\sqrt{2}-1)}(\mathcal{B}_{\text{Child}}-\mathcal{B}_{\text{Parent}})\times 100\%.    
\end{align}
Manifestly, the gain will be $100\%$ when by wiring a nonlocal correlation with arbitrary small CHSH violation, the distilled correlation achieves Tsirelson's bound -- the maximum CHSH value in $\mathcal{Q}$ \cite{Cirelson1980}. We then obtain that under the OR-AND wiring a quantum Hardy correlation can yield Tsirelson gain at-most $20.9\%$.  

While applying multi-copy OR-AND protocol it turns out that the optimal distillation of Hardy's success is obtained with a threshold number of initial boxes, and the success gets decreased if more number of initial boxes are considered. Our next proposition provides an exact expression for the optimal number of initial Hardy correlation required for maximal distillation.
\begin{proposition}\label{prop1}
A no-signaling Hardy correlation of the form of Eq.(\ref{q-hardy}) yields Hardy success $p^{(n)}_{_{\text{Hardy}}}=(\frac{c_0}{2}+c_1)^{n}-c_1^{n}$, when its $n$-copy is wired under OR-AND protocol. The optimal value of distilled Hardy success is given by $p^{opt}_{_{\text{Hardy}}}=\max \{ p^{(\mathrm{N})}_{_{\text{Hardy}}} ,~p^{(\mathrm{N}+1)}_{_{\text{Hardy}}}\}$, where
\begin{align}
\mathrm{N}:=\left\lfloor\log\left(\frac{\log{c_1}}{\log{\frac{c_0}{2}+c_1}}\right)/\log\left(\frac{\frac{c_0}{2}+c_1}{c_1}\right)\right\rfloor. 
\label{Ncopy}
\end{align}
\end{proposition}
Proof provided in Supplemental. In the next, we rather consider quantum nonlocal correlation, not necessarily in Hardy form, to establish an even higher Tsirelson gain through the OR-AND wiring.
\begin{theorem}\label{theo2}
Starting with a quantum correlation with arbitrarily small CHSH nonlocality OR-AND wiring can yield Tsirelson gain up to $(\approx)~39.75\%$.
\end{theorem}
Proof of Theorem \ref{theo2} follows similar arguments as of Theorem \ref{theo1}. However, for the sake of completeness detailed proof is provided in supplementary file.  

Theorem \ref{theo2} has many interesting implications for information processing tasks, wherein higher CHSH violation is desirable for higher performance of the protocols. For instance, the amount of certifiable randomness as obtained in \cite{Pironio2010} monotonically scales with the degree of violation of CHSH inequality. On the other hand, the authors in \cite{Pappa2015} come up with a conflicting interest Bayesian game where the payoff in correlated equilibrium strategy increases linearly with the amount of CHSH violation (see also \cite{Roy2016}). More recently, the authors in \cite{Frenkel2022} proposed a communication task where preshared entanglement between sender and receiver is shown to enhance the communication utility of a perfect classical communication channel. As it turns out payoff of this task is also a linear function of the value of CHSH expression \cite{Patra2022}.\\

By now, an observant reader might have already noticed that the local box $\mathrm{P}_{L_1}$ plays crucial role in the proof of Theorem \ref{theo1} and Theorem \ref{theo2}. We will now use this observation to prove a generic result as formalized in the following theorem.
\begin{theorem}\label{theo3}
CHSH nonlocality of any no-singling correlation of the form $\tilde{\mathrm{C}}(\lambda)=\lambda\mathrm{C}+(1-\lambda)P_{L_1}$, where $0<\lambda\leq 1$ and $\mathrm{C}\in\mathrm{ConvexHull}~\{\mathrm{P}_{NL},\mathrm{P}_{L_i}~|~i\in\{1,\cdots,8\}\}$  can be distilled through OR-AND wiring by choosing the values of $\lambda$ sufficiently small. Furthermore,  
$2$-copy OR-AND distillation is successful for all the $\tilde{C}(\lambda)$ correlation boxes whenever $\lambda <\frac{2}{3}c_0$; where $c_0$ is the $\mathrm{P}_{NL}$ fraction in $\mathrm{C}$.
\end{theorem}
\begin{proof}
Given two correlations $\chi_1,\chi_2\in\mathcal{NS}$, let $\mathcal{W}[\chi_1,\chi_2]$ denotes the resulting correlation obtained under OR-AND wiring, where $\mathcal{W}[\chi,\chi]\equiv\chi^{(2)}$. We, therefore, have
\begin{align*}
\tilde{\mathrm{C}}^{(2)}(\lambda)&=\lambda^2\mathrm{C}^{(2)}+\lambda(1-\lambda)\left\{\mathcal{W}\left[\mathrm{C},P_{L_1}\right]+\mathcal{W}\left[P_{L_1},\mathrm{C}\right]\right\}\\
&~~~~~~~+(1-\lambda)^2\mathcal{W}\left[P_{L_1},P_{L_1}\right].
\end{align*}
A straightforward calculation yields, $ \mathcal{W}\left[\mathrm{C},P_{L_1}\right]=\mathrm{C}=\mathcal{W}\left[P_{L_1},\mathrm{C}\right]$, that further result in 
\begin{align*}
\tilde{\mathrm{C}}^{(2)}(\lambda)=\lambda^2\mathrm{C}^{(2)}+2\lambda(1-\lambda)\mathrm{C}+(1-\lambda)^2P_{L_1}. \end{align*}
While CHSH value of the box $P_{L_1}$ is $2$, let the CHSH value of the boxes $\mathrm{C}$ and $\mathrm{C}^{(2)}$ be denoted as $\mathcal{K}~(>2)$ and $\mathcal{K}^{(2)}$ respectively. Then, the CHSH value of the correlations $\tilde{\mathrm{C}}(\lambda)$ and $\tilde{\mathrm{C}}^{(2)}(\lambda)$ can be expressed as,
\begin{align*}
\mathcal{B}(\lambda)&=\lambda~\mathcal{K}+2~(1-\lambda),\\
\mathcal{B}^{(2)}(\lambda)&=\lambda^2~\mathcal{K}^{(2)}+2~\lambda~(1-\lambda)\mathcal{K}+2~(1-\lambda)^2.
\end{align*}
A successful distillation demands $\mathcal{B}^{(2)}(\lambda)>\mathcal{B}(\lambda)$, implying $(\mathcal{K}-2)+(\mathcal{K}^{(2)}-2\mathcal{K}+2)\lambda>0$, which can be guaranteed by choosing the values for $\lambda$ accordingly. This completes first part of the proof. A bit more calculation yeilds the quantitative bound $\lambda<\frac{2}{3}c_0$, on the radius of the eight dimensional Ball centered at point ${\rm P}_{L_1}$ assuring $2$-copy OR-AND distillation such that any nonlocal no-signaling correlation, be it quantum or post-quantum, choosen from a nonzero-measure sector of the Ball can be distilled (in the full eight dimensions). We provide the detailed calculation in the appended Supplemental Material .   
\end{proof}
Theorem \ref{theo3} has profound topological implication. It establishes that the sets of no-signalling as well as quantum correlations allowing nonlocality distillation have non-zero measure in the full eight dimensional correlation space. Furthermore, it should be mentioned that the correlation box $\mathrm{P}_{L_1}$ appearing in Theorem \ref{theo3} is not any special local deterministic box: the result holds also for all the remaining $15$ local deterministic boxes on suitable relabeling of the OR-AND wiring. 

While the studies in nonlocality distillation of quantum correlations are mostly limited to the 2-2-2 Bell scenario, Theorem \ref{theo3} opens up an avenue to study the same in a general N-M-K scenario that involves N spatially separated parties each performing M different measurements with K outcomes. It is not hard to find the extreme local boxes in such a general scenario. Now if for such an extreme box $P_{L}$ we obtain a wiring $\mathrm{W}$ such that $\mathrm{W}[P_{L},\mathrm{X}]=X=\mathrm{W}[\mathrm{X},P_{L}]$ for any N-M-K no-signaling correlation $\mathrm{X}$, then it results in a generalization of Theorem \ref{theo3} in the N-M-K scenario. This consequently will imply that quantum  correlations allowing nonlocality distillation have non-zero measure even in this general scenario.

{\it Distillation of post-quantum nonlocality.--} We now proceed to show that OR-AND wiring also has an important proviso in ruling out (unphysical) post-quantum correlations. Several techniques are there to establish post-quantumness of a given correlations. For instance, isotropic no-signaling correlations yielding CHSH value more than $4\sqrt{2/3}$ violate the principle of nontrivial communication complexity \cite{Brassard2005} (see also \cite{vanDam2005,Brassard2005(1),Buhrman2010}), thus demarcating such correlations as unphysical. Furthermore, any NS correlating with CHSH value more than Tsirelson bound violates the principle of information causality \cite{Pawlowski2009,Miklin2021}. It has also been shown that a correlation might not violate these principles by its own, but after distillation the resulting correlation violates such a principle, which, in turn, establish unphysicality of the original correlation \cite{Brunner2009} (see also \cite{Eftaxias2022}).

Interestingly, the OR-AND wiring becomes useful to establish post-quantum nature of a correlation. Let us consider the following NS correlation 
\begin{align}
\mathrm{H}_{NS}&=0.1~\mathrm{P}_{NL}+0.85~\mathrm{P}_{L_1}\nonumber\\
&+0.01~(\mathrm{P}_{L_2}+\mathrm{P}_{L_3}+2~\mathrm{P}_{L_4}+\mathrm{P}_{L_5}), \end{align}
which exhibits Hardy's nonlocality with success probability $p_{_{\text{Hardy}}}=0.05$. However, correlation $\mathrm{H}^{(8)}_{NS}$ obtained by distilling $8$ copies of $\mathrm{H}_{NS}$ has the Hardy success $p^{opt}_{_{\text{Hardy}}}=0.15797$, which ensures that the boxes $\mathrm{H}^{(8)}_{NS}$ and $\mathrm{H}_{NS}$ are post-quantum. As it turns out the correlation $\mathrm{H}_{NS}$ neither violates known necessary condition for respecting the principal of nontrivial communication complexity nor for the principle of information causality. However, the considered example violates the macroscopic locality principle \cite{Navascues2009} (see Supplemental). That being said, we do point out that checking membership to the NPA hierarchy can become computationally expensive,
particularly at higher orders of the hierarchy, while the distillation criteria is far more tractable (see Supplemental).

{\it Discussion.--} Distillation, the process of extracting a desirable substance in pure form from a source of impure mixture through heating and other means, has an ancient history. Quite interestingly, during recent past, the idea finds novel applications in quantum information theory, where one aims to obtain fewer number of higher resourceful states starting with larger number of lesser resourceful states under free operations~\cite{Chitambar}. Some canonical examples are: (i) for our present study, the resource theory of Bell-nonlocal boxes~\cite{Vicente,Wolfe}, or (ii) the well known protocols for quantum entanglement distillation, where many copies of mixed entangled states are distilled into pure form under local quantum operations and classical communications~\cite{Bennett1996(1)}.

In this letter, we have established a generic approach for distillation of nonlocal correlations arising in quantum mechanics. This problem is of utmost importance as Bell nonlocal correlations are ubiquitous in device independent protocols -- more the nonlocality more the utility. Interestingly, we come up with an elegant protocol, the OR-AND wiring, that distils nonlocality in quantum correlations with high efficiency. In the simplest bipartite scenario, in stark distinction with the results reported prior to our work \cite{Forster2009,Brunner2009,Short2009,Brunner2011,Forster2011,HoyerJibran2010,Jibran2012,HoyerJibran2013, Ebbe2014,Tuziemski2015,Beigi2015,Brassard2015,Brito2019,Eftaxias2022,Allcock2009,Lang2014}, our protocol establishes that, within the set of full eight dimensional correlation space, the distillable quantum as well as no-signaling nonlocal correlations form subsets of non-zero measures; i.e., sector of open balls of a specified radius centered at local deterministic correlations. Moreover, by considering correlations arbitrarily close to local deterministic points, applying our protocol, with optimal number of copies, one can distill nonlocality by a significant amount both for the quantum as well as post-quantum non-signaling correlations. As for future, it would be interesting to explore the full potential of our generic framework proposed here in distilling quantum nonlocal correlations. In particular, obtaining some bound on the relative volume of the quantum correlations in the correlation space that can be distilled under OR-AND wiring would be interesting. Furthermore a generalization of this protocol for higher input-output as well as in multiparty scenario might be of great use.  
\begin{acknowledgments}
We gratefully acknowledge fruitful discussions and feedback from Guruprasad Kar, Giorgos Eftaxias, Roger Colbeck, Kieran Flatt and Joonwoo Bae at different stages of this work. M. B. acknowledges funding from the National Mission in Interdisciplinary Cyber-Physical systems from the Department of Science and Technology through the I-HUB Quantum Technology Foundation (Grant No. I-HUB/PDF/2021-22/008), support through the research grant of INSPIRE Faculty fellowship from the Department of Science and Technology, Government of India, and the start-up research grant from SERB, Department of Science and Technology (Grant No. SRG/2021/000267). A. R. is supported by National Research Foundation of Korea (NRF-2021R1A2C2006309), Institute of Information Communications Technology Planning and Evaluation (IITP) Grant (Grant No. RS-2023-00229524, the ITRC Program/IITP-2023-2018-0-01402).
\end{acknowledgments}

\maketitle
\onecolumngrid    
\section{SUPPLEMENTAL}
\section{Appendix A: Analysis of Theorem 1}\label{appendix-A}
Before starting the main analysis of the Theorem-1, first note that any NS correlation $\mathrm{C}\in\mathcal{NS}$ can be represented in the following matrix form: 
\begin{align}
\mathrm{C}\equiv\begin{blockarray}{ccccc}
xy/ab& ~~~~00~~~~ & ~~~~01~~~~ & ~~~~10~~~~ & ~~~~11~~~~ \\\\
\begin{block}{c(cccc)}
00 & p(00|00) & p(01|00) & p(10|00) & p(11|00) \\\\
01 & p(00|01) & p(01|01) & p(10|01) & p(11|01) \\\\
10 & p(00|10) & p(01|10) & p(10|10) & p(11|10) \\\\
11 & p(00|11) & p(01|11) & p(10|11) & p(11|11) \\
\end{block}
\end{blockarray}
\end{align}
where $p(ab|xy)$ denotes the probability of obtaining outcome $a$ on Alice's side and outcome $b$ on Bob's side given there inputs $x$ and $y$ respectively; $a,b,x,y\in\{0,1\}$. Consider the Bell-CHSH inequality defined in the main text of the paper, {\it i.e.},
\begin{align}
\mathcal{B}\equiv \langle00\rangle-\langle01\rangle-\langle10\rangle-\langle11\rangle \leq 2.
\end{align}
There is only one PR-box that violates the considered inequality, {\it i.e.}, yields CHSH value $\mathcal{B}=4$; and there are exactly eight extremal local boxes for which $\mathcal{B}=2$. These boxes are given by,
\begin{align*}
\mathrm{P}_{NL}\equiv\mathrm{P}^{110}_{NL}
&=\begin{pmatrix}
\frac{1}{2} & 0 & 0 & \frac{1}{2} \\
0 & \frac{1}{2} & \frac{1}{2} & 0 \\
0 & \frac{1}{2} & \frac{1}{2} & 0 \\
0 & \frac{1}{2} & \frac{1}{2} & 0 \\
\end{pmatrix};~~~
\mathrm{P}_{L_1}\equiv\mathrm{P}^{0001}_{L_1}=\begin{pmatrix}
0 & 1 & 0 & 0 \\
0 & 1 & 0 & 0\\
0 & 1 & 0 & 0\\
0 & 1 & 0 & 0\\
\end{pmatrix};~~~
\mathrm{P}_{L_2}\equiv\mathrm{P}^{0100}_{L_2}=\begin{pmatrix}
0 & 0 & 1 & 0 \\
0 & 0 & 1 & 0\\
0 & 0 & 1 & 0\\
0 & 0 & 1 & 0\\
\end{pmatrix};\\\\
\mathrm{P}_{L_3}\equiv\mathrm{P}^{0111}_{L_3}&=\begin{pmatrix}
0 & 0 & 0 & 1 \\
0 & 0 & 1 & 0\\
0 & 0 & 0 & 1\\
0 & 0 & 1 & 0\\
\end{pmatrix};~~~~~
\mathrm{P}_{L_4}\equiv\mathrm{P}^{1101}_{L_4}=\begin{pmatrix}
0 & 0 & 0 & 1 \\
0 & 0 & 0 & 1\\
0 & 1 & 0 & 0\\
0 & 1 & 0 & 0\\
\end{pmatrix};~~~
\mathrm{P}_{L_5}\equiv\mathrm{P}^{1111}_{L_5}=\begin{pmatrix}
0 & 0 & 0 & 1 \\
0 & 0 & 1 & 0\\
0 & 1 & 0 & 0\\
1 & 0 & 0 & 0\\
\end{pmatrix};\\\\
\mathrm{P}_{L_6}\equiv\mathrm{P}^{1000}_{L_6}&=\begin{pmatrix}
1 & 0 & 0 & 0 \\
1 & 0 & 0 & 0\\
0 & 0 & 1 & 0\\
0 & 0 & 1 & 0\\
\end{pmatrix};~~~~~
\mathrm{P}_{L_7}\equiv\mathrm{P}^{0010}_{L_7}=\begin{pmatrix}
1 & 0 & 0 & 0 \\
0 & 1 & 0 & 0\\
1 & 0 & 0 & 0\\
0 & 1 & 0 & 0\\  
\end{pmatrix};~~~
\mathrm{P}_{L_8}\equiv\mathrm{P}^{1010}_{L_8}=\begin{pmatrix}
1 & 0 & 0 & 0 \\
0 & 1 & 0 & 0\\
0 & 0 & 1 & 0\\
0 & 0 & 0 & 1\\
\end{pmatrix}.
\end{align*}
In Theorem 1 we have showed that OR-AND protocol preserves the quantum Hardy structure and moreover it can distil non-local strength of Hardy correlations. Here, we analyse this theorem in detail.

First we explicitly show that Hardy structure is preserved by the OR-AND wiring. For this, consider the quantum Hardy correlation 
\begin{align}
\mathrm{H}_Q&=c_0\mathrm{P}_{NL}\!+\!\sum_{i=1}^5c_i\mathrm{P}_{L_i};~~c_i>0~~\forall~i,~~\mbox{and}~~\sum_{i=0}^5 c_i=1.\label{HQ1}
\end{align}
With little algebra, it can be shown that OR-AND wiring applied to two copies of $\mathrm{H}_Q$ results in a correlation
\begin{align}
\mathrm{H}^{(2)}_Q=&2\left(\left(\frac{c_0}{2}+c_1\right)^2-c_1^2\right) \mathrm{P}_{NL}+c_1^2\mathrm{P}_{L_1}+\left(c_0\left(1-\frac{c_0}{2}\right)+2c_2\left(1-\frac{c_2}{2}\right)-c_0\left(c_1+c_2\right)\right)\mathrm{P}_{L_2}\nonumber\\
&+c_3\left(2-c_0-c_3-2c_2\right)\mathrm{P}_{L_3}+c_4\left(c_0+2c_1+c_4\right)\mathrm{P}_{L_4}+c_5 \left(c_0 + 2 c_1 + 2 c_4 + c_5\right)\mathrm{P}_{L_5}.\label{HQ2}
\end{align}
Eq.(\ref{HQ2}) is clearly in the form Eq.(\ref{HQ1}) (i.e., both are convex mixture of same set of one nonlocal and five local boxes), and therefore the correlation resulting from wiring is a Hardy non-local correlation, establishing that OR-AND protocol preserves Hardy structure of quantum Hardy correlations. It is also clear from the expression that 2-copy OR-AND wiring distils nonlocality whenever $\frac{c_0}{2}+2c_1>1$.
\begin{figure}[b!]
\centering
\includegraphics[width=0.45\textwidth]{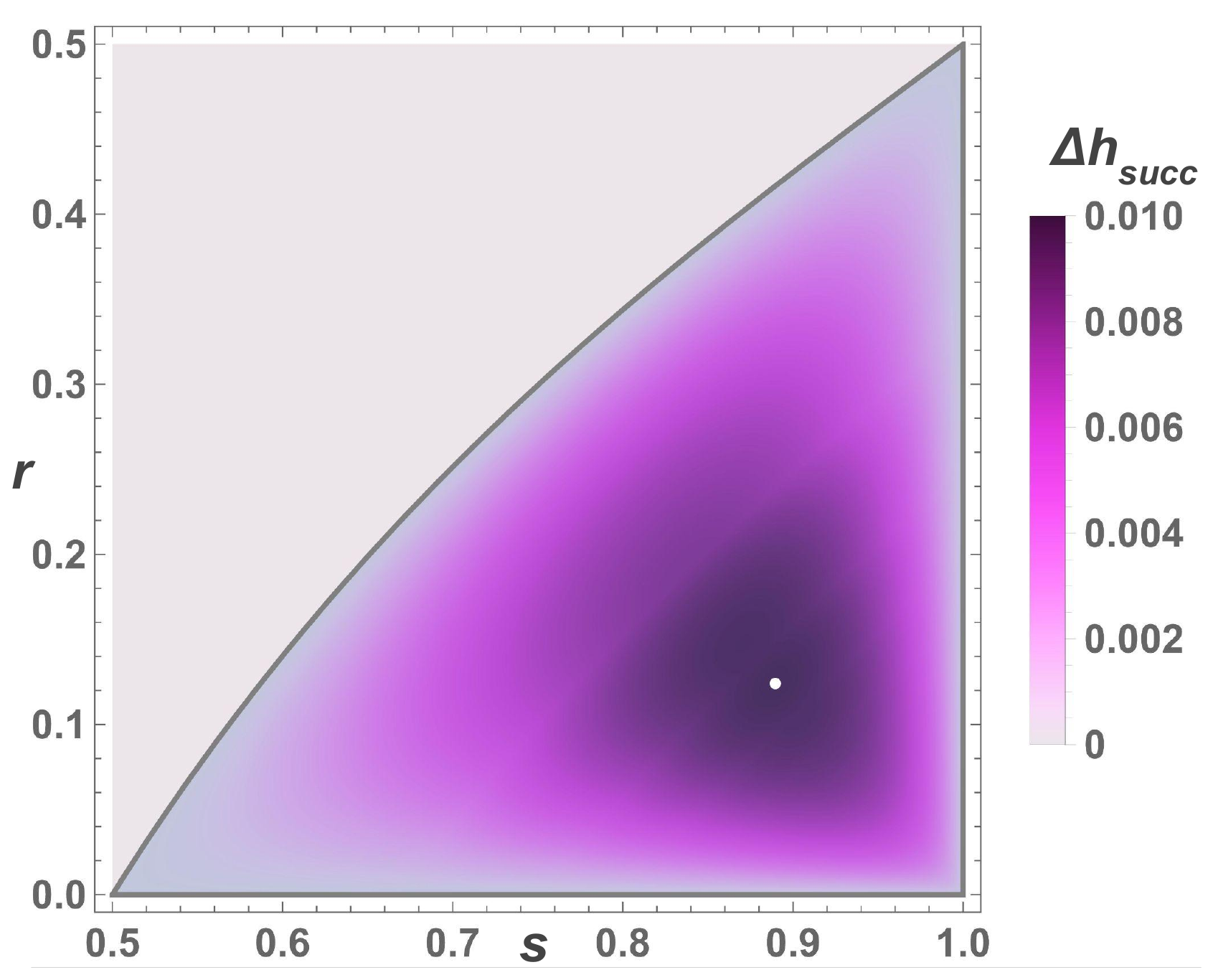}
\caption{The variation of distillation gap with $r$ and $s$. The coloured region satisfies Eq.(\ref{region}) and the distillation gap of an $(r,s)$ point is mapped by the color as shown in the colorbar. The white point denotes the correlation with maximal distillation gap.}
\label{Fig:A1}
\end{figure}

Now, let us turn to the distillability of quantum Hardy correlations under multi copy OR-AND wiring. In general,  here we consider a two parameter family of quantum Hardy correlations which also includes the maximum quantum Hardy correlation used in Theorem 1. The family of such correlations is as follows, 
\begin{align}
\mathrm{H}_{\mathcal{Q}}^{(r,s)}&\equiv\begin{blockarray}{ccccc}
xy/ab& 00 & 01 & 10 & 11 \\
\begin{block}{c(cccc)}
00 & \frac{(1-r)r(1-s)s}{1-rs} & \frac{(1-r)^2s}{1-rs} & (1-r)(1-s) & r \\\\
01 & 0 & (1-r)s & \frac{1-s}{1-rs} & \frac{(1-r)rs^2}{1-rs} \\\\
10 & 0 & s & \frac{(1-r)(1-s)}{1-rs} & \frac{r(1-s)^2}{1-rs} \\\\
11 & \frac{r(1-s)s}{1-rs} & \frac{(1-r)s}{1-rs} & 1-s & 0 \\
\end{block}
\end{blockarray}\label{qhardy22}
\end{align}
where, $r,s\in [0,1]$. The measurements performed by Alice and Bob are given by 
\begin{subequations}
\begin{align}
\mbox{Alice's measurements}&:~\left\{x_0\equiv\{\ket{0}\bra{0},\ket{1}\bra{1}\},~~~~~x_1\equiv\{\ket{u_0}\bra{u_0},\ket{u_1}\bra{u_1}\}\right\},\\
\mbox{Bob's measurements}&:~\left\{y_0=\{\ket{0}\bra{0},\ket{1}\bra{1}\},~~~~~y_1=\{\ket{v_0}\bra{v_0},\ket{v_1}\bra{v_1}\}\right\},
\end{align}\label{mesh}
\end{subequations}
where, the first term in the brackets corresponds to the outcome $0$ and second term the outcome $1$, and 
\begin{align}
\vert u_0\rangle&:=\mathbf{C}_\alpha\vert 0\rangle + e^{\iota \phi}\mathbf{S}_\alpha\vert 1\rangle,~~~\vert u_1\rangle:=-\mathbf{S}_\alpha\vert 0\rangle + e^{\iota \phi}\mathbf{C}_\alpha\vert 1\rangle,\\
\vert v_0\rangle&:=\mathbf{C}_\beta\vert 0\rangle + e^{\iota \xi}\mathbf{S}_\beta\vert 1\rangle,~~~~\vert v_1\rangle:=-\mathbf{S}_\beta\vert 0\rangle + e^{\iota \xi}\mathbf{C}_\beta\vert 1\rangle,\\
\mbox{with}~~\mathbf{C}_z&:=\cos(z/2),~~\mathbf{S}_z:=\sin(z/2)~~ \mbox{and}~~\alpha,\beta\in[0, \pi]~~\&~\phi, \xi \in [0, 2\pi).\nonumber
\end{align}
The two-qubit state on which the measurements are performed, yielding the correlation (\ref{qhardy22}), is given by 
\begin{align}
\ket{\psi}_{_{Hardy}}:=\frac{\ket{u_0v_0}+\mathbf{W}_{\alpha}\ket{u_1v_0}+\mathbf{W}_{\beta}\ket{u_0v_1}}{\sqrt{1+\mathbf{W}_{\alpha}^2+\mathbf{W}_{\beta}^2}};~~~~\mathbf{W}_{z}:=\cot(z/2)
\end{align}
where $r:=1-\mathbf{S}^2_{\alpha}\mathbf{S}^2_{\beta}$ and $s:=r^{-1}\mathbf{C}^2_{\alpha}$ \cite{Rai2022}. It can be verified that maximal Hardy non-locality is achieved when $r=s=\frac{1}{2}\left(\sqrt{5}-1\right)$, resulting in Eq.(3) of the main text. Correlation $\mathrm{H}_{\mathcal{Q}}^{(r,s)}$ can be distilled using OR-AND protocol {\it if and only if}
\begin{align}
r^2(s+s^2)-r(s^2+2s)+(2s-1)>0.  \label{region}
\end{align}
Furthermore, it turns out that the region of distillability remains the same for the $n$-copy wiring. One can take a step further, and quantify the distillability of these correlations under the $n$-copy OR-AND protocol using distillation gap, defined as $\Delta h{_\text{succ}}:=p^{(N^{\text{opt}})}_{_\text{Hardy}}-p_{_\text{Hardy}}$, where $N^{\text{opt}}$ is the optimal number of copies obtained from Proposition 1. The resulting variation of distillation gap with $r$ and $s$ is shown in Fig.\ref{Fig:A1}. The colored region in the density plot satisfies Eq.(\ref{region}) and the distillation gap of a point is mapped by the color as indicated by the colorbar. The white point in the graph denotes the optimal distillation of approximately $0.0101896$, obtained when $4$ copies of the quantum correlation $\mathrm{H}_{\mathcal{Q}}^{(r,s)}$ identified by $r\approx0.1241$ and $s\approx0.8896$ is used.

Interestingly, as shown in the main text, much higher distillation gap can be obtained when one considers convex mixtures of the correlations given in Eq.(\ref{qhardy22}) and $\mathrm{P}_{L_1}$. More strikingly, high amount of nonlocality can be distilled from mixtures with vanishingly small contribution from quantum Hardy correlations. For example, consider the correlation from Eq.(5) of the main text given by,
\begin{align*}
\tilde{\mathrm{H}}_Q(\lambda):=\lambda\mathrm{H}^{\max}_Q+(1-\lambda)\mathrm{P}^{[1]}_L,~~\lambda\in(0,1].
\end{align*}
The Hardy success of this correlation is $p_{_\text{Hardy}}(\lambda)=\lambda k_1=\frac{\lambda}{2}(5\sqrt{5}-11)$. Now, after using the $n$-copy OR-AND protocol, the resulting Hardy success is given by
\begin{align*}
p^{(n)}_{_\text{Hardy}}(\lambda)=(c_0/2+c_1)^n-c_1^n=\left(\left(\sqrt{5}-3\right) \lambda +1\right)^n-\left(\frac{\sqrt{5}}{2} \left(\sqrt{5}-3 \right) \lambda +1\right)^n.
\end{align*}
Now, assuming very small $\lambda$ and Taylor expanding around $\lambda=0$, it turns out that the leading order term in the expression for optimal number of copies (see Proposition 1) is of the order $\lambda^{-1}$. Hence, for very small values of $\lambda$, we can ignore the floor function arising in $\mathrm{N}^{opt}$. Thus, we have the following for $ p^{(\mathrm{N}^{opt})}_{_\text{Hardy}}$ in the small $\lambda$ limit,
\begin{align}\label{QMaxNopt}
p^{(\mathrm{N}^{opt})}_{_\text{Hardy}}(\lambda)&=4^{(2+\sqrt{5})}~ 5^{(-\frac{5}{2}-\sqrt{5})} \left(\sqrt{5}-2\right)-4^{(1+\sqrt{5})}~5^{(-2-\sqrt{5})} \left(\sqrt{5}-3\right)\log \left(\frac{5}{4}\right)\lambda+O\left(\lambda ^2\right)\nonumber\\
&\approx 0.0410237+0.0165601 \lambda+O\left(\lambda^2\right)
\end{align}
Eq.(\ref{QMaxNopt}) tells that even as $\lambda$ approaches $0$, corresponding to vanishingly small initial non-locality, the optimal use of the $n$-copy OR-AND can distil a high nonlocality of approximately $0.0410237$. In fact, this increases linearly with $\lambda$.
\begin{figure}[t!]
\centering
\includegraphics[width=0.4\textwidth]{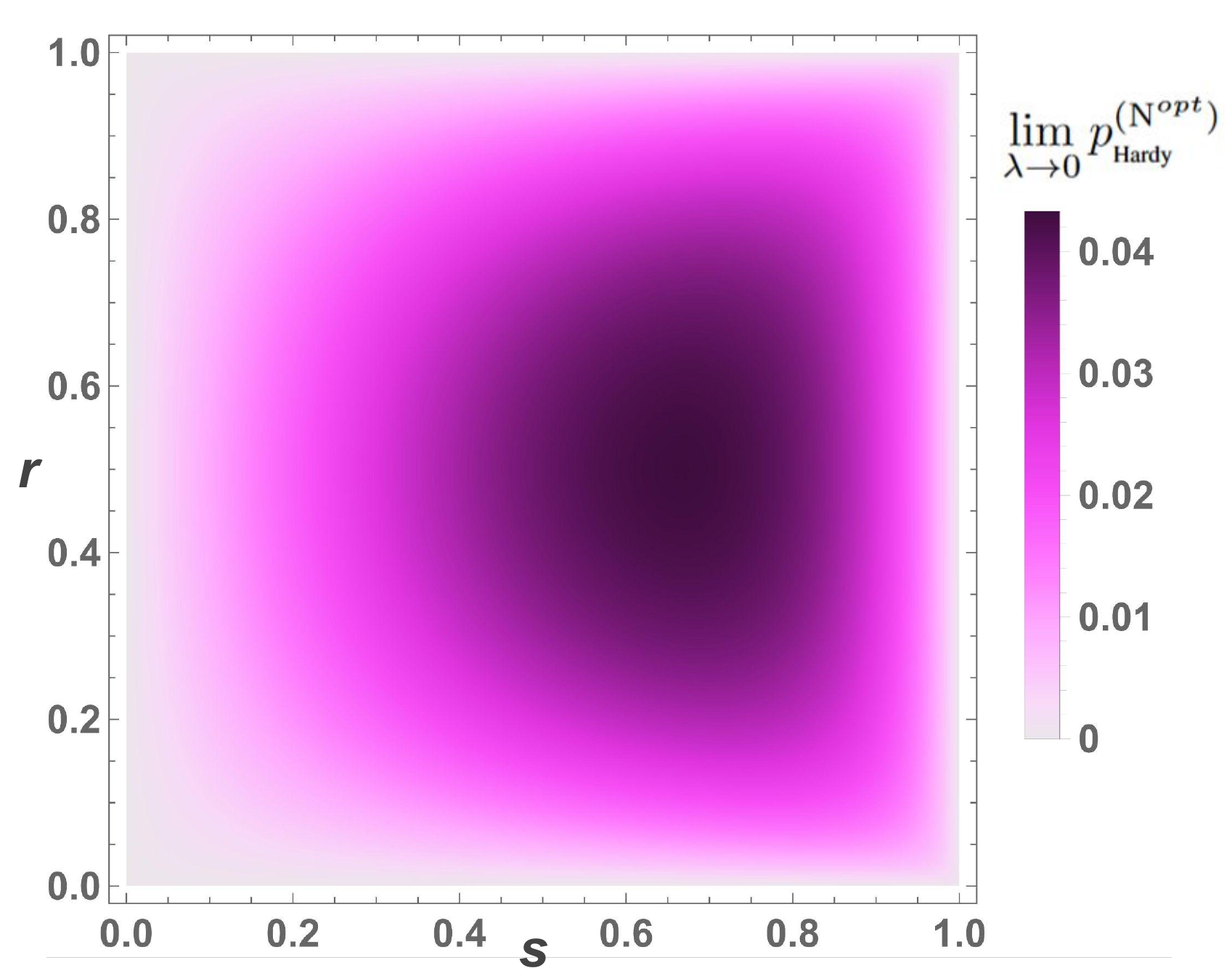}
\caption{The variation of distilled Hardy success in the $\lambda\rightarrow 0$ with $r$ and $s$. The Hardy success of an $(r,s)$ point is mapped by the color as shown in the colorbar.}
\label{Fig:A2}
\end{figure}

In the above, we have considered the mixture of optimal quantum Hardy correlation $\mathrm{H}^{\max}_Q$ with $\mathrm{P}_{L_1}$. We can consider the same scenario for arbitrary $\mathrm{H}^{(r,s)}_Q$. Then again, in the small $\lambda$ limit, one can follow similar steps to arrive at the following expression for the optimal Hardy success in the limit of vanishing $\lambda$,
\begin{align}\label{QMaxNoptrs}
\lim_{\lambda\rightarrow 0}p^{(\mathrm{N}^{opt})}_{_\text{Hardy}}(r,s,\lambda)&=\frac{(1-r) r (1-s) s \left(\frac{1-s+(1-r) r s}{1-s+(1-r) r s^2}\right)^{\left(1-\frac{1}{(1-r) rs}-\frac{1}{1-s}\right)}}{1-s+(1-r) r s}.
\end{align}
A density plot of Eq.(\ref{QMaxNoptrs}) is shown in Fig.\ref{Fig:A2}, where the distilled Hardy success of the point $(r,s)$ is indicated by the color of the point. Moreover, one can find the optimal value of Eq.(\ref{QMaxNoptrs}) to be approximately $0.0433049$, attained when $r=\frac{1}{2}$ and $s=\frac{2}{3}$. This is interesting, since it means that the mixture of a weakly non-local Hardy correlation with local noise can distil more nonlocality than a mixture with the optimal quantum nonlocal Hardy correlation and the same local noise. This again reveals the nontrivial structure of nonlocality distillation. Finally, note that despite the ability of the OR-AND protocol to distill high Hardy success from a mixture of $\mathrm{H}^{(r,s)}_Q$ and $\mathrm{P}_{L_1}$, one can numerically verify that the distilled success of this mixture is never greater than the distilled success of  $\mathrm{H}^{(r,s)}_Q$ itself.

\section{Appendix B: Proof of Proposition 1}
\begin{proof}
Condition for $\mathrm{n}$ copy distillation demands that at least two copy distillation must be successful, i.e., $\mathrm{n}\geq 2$. This imposes the condition $c_1>\frac{1}{2}-\frac{c_0}{2}$ (along with other conditions $c_0\geq 0$, $c_1\geq 0$, and $c_0+c_1\leq 1$ on the coefficients $c_0$ and $c_1$ of the respective boxes $\mathrm{P}_{NL}$ and $\mathrm{P}_{L_1}$). Let us consider a continues and smooth function $f(x)=\left(\frac{c_0}{2}+c_1\right)^{x}-\left(c_1\right)^{x}$ obtained by substituting the number of copies $\mathrm{n}$ with a continuous variable $x\in[1,\infty)$ in the expression for Hardy's success probability $p^{n}_{_{\text{Hardy}}}=(\frac{c_0}{2}+c_1)^{n}-c_1^{n}$. The derivative $\frac{df(x)}{dx}$ takes the value zero at exactly one point given by 
\begin{equation}
  x^{\ast}=\log\left(\frac{\log{c_1}}{\log{\frac{c_0}{2}+c_1}}\right)/\log\left(\frac{\frac{c_0}{2}+c_1}{c_1}\right).
\end{equation}
 Then, one can conclude that the critical point $x^{\ast}$ is the point of global maxima of $f(x)$ because due to at least two copy distillation condition $f(x)$ must be an increasing function in the beginning, then it achieves the global maxima at $x^{\ast}$ and starts decreasing when $x> x^{\ast}$. Let us define $\mathrm{N}=\lfloor x^{\ast}\rfloor$, then since $f(x^{\ast})\geq f(x)$ for all $x\in[1,\infty)$ it easily follows that the optimal value of distilled Hardy success is given by $p^{opt}_{_{\text{Hardy}}}=\max \{ p^{\mathrm{(N)}}_{_{\text{Hardy}}} ,~p^{(\mathrm{N}+1)}_{_{\text{Hardy}}}\}$. The optimal number of copies of initial boxes yielding maximum distillation of Hardy success, i.e., $\mathrm{N}^{opt}$ is either $\mathrm{N}$ or $\mathrm{N}+1$; the one which gives maximum Hardy success. This completes the proof.
\end{proof}

\section{Appendix C: Proof and Analysis of Theorem 2}
\begin{proof}
	Consider a one parameter family of quantum correlations given by
	\begin{align}\label{bql}
		\tilde{\mathrm{B}}_Q(\lambda)&:=\lambda \mathrm{B}^{\max}_Q+(1-\lambda)\mathrm{P}_{L_1},~~\lambda\in(0,1],~~\mbox{where},\\\nonumber
		\mathrm{B}_{Q}^{\max}&:=(\sqrt{2}-1)\mathrm{P}_{NL}+\frac{1}{4}\left(1-\frac{1}{\sqrt{2}}\right)\sum_{i=1}^8\mathrm{P}_{L_i}. 
	\end{align}
	Note that the correlation $\mathrm{B}_{Q}^{\max}$ saturates the Tsirelson's bound in $\mathcal{Q}$. Manifestly, $\tilde{\mathrm{B}}_Q(\lambda)$ yields CHSH value  $\mathcal{B}(\lambda)=2+(2\sqrt{2}-2)\lambda$. After $2$-copy OR-AND wiring the resulting correlation $\tilde{\mathrm{B}}^{(2)}_Q(\lambda)$ attains CHSH value $\mathcal{B}^{(2)}(\lambda)=2+2(2\sqrt{2}-2)\lambda+(11/4-3\sqrt{2})\lambda^2$. Distillation is successful whenever $\lambda<\frac{8}{167}(13-\sqrt{2})\approx0.555$. While appropriate choice of $\lambda$ yields maximal Tsirelson gain $(\approx)~13.9\%$ with $2$-copy distillation of the correlation $\tilde{\mathrm{B}}_Q(\lambda)$, multi-copy OR-AND wirings can yield a maximal gain $(\approx)~39.75\%$ and is obtained via distillation of parent correlation with arbitrary small value of $\lambda$ (see Supplemental). 
\end{proof}
In the proof to Theorem 2, we have showed that a mixture containing $\mathrm{P}_{L_1}$ and arbitrarily small amount of $\mathrm{B}^{\text{max}}_{Q}$ can be distilled using the OR-AND protocol. Now, we numerically show that by using the optimal number of copies, we can distill to a high CHSH value of 2.32.

For this, first note that, given an arbitrary correlation $\mathrm{P}\equiv\{p(ab|xy)\}$ the resulting correlation after n-copy OR-AND protocol is given by,
\begin{subequations}\label{ncopy}
\begin{align}
    p^{(n)}(00|xy)&=(p(00|xy)+p(01|xy))^n-p(01|xy)^n,\\
    p^{(n)}(01|xy)&=p(01|xy)^n, \\
    p^{(n)}(11|xy)&=(p(11|xy)+p(01|xy))^n-p(01|xy)^n,\\
    p^{(n)}(10|xy)&=1-p^{(n)}(00|xy)-p^{(n)}(01|xy)-p^{(n)}(11|xy).
\end{align}
\end{subequations}
Eq.(\ref{ncopy}) can be easily verified using mathematical induction starting from the 2-copy OR-AND protocol. Using Eq.(\ref{ncopy}), we arrive at the following expression for the distlled CHSH value ($\mathcal{B}^n(\lambda)$) of the correlation $\tilde{\mathrm{B}}_Q(\lambda)$ using the n-copy OR-AND protocol:
\begin{align}
\label{ncopychsh}
\mathcal{B}^{(n)}(\lambda)&=2-8 \left(1-\frac{\lambda }{2}\right)^n+12 \left(1-\frac{1}{8} \left(6-\sqrt{2}\right)
\lambda\right)^n-4 \left(1-\frac{1}{8} \left(6+\sqrt{2}\right) \lambda \right)^n.
\end{align}
Eq.(\ref{ncopychsh}) can now be used to estimate the maximum CHSH value that can be distilled in the $\lambda\rightarrow 0$ limit. However, this scenario is more difficult compared to Appendix:A since we do not have an analytical expression for the optimal number of copies for maximal distillation of the CHSH value. 
\begin{figure}[b!]
\centering
\includegraphics[width=0.4\textwidth]{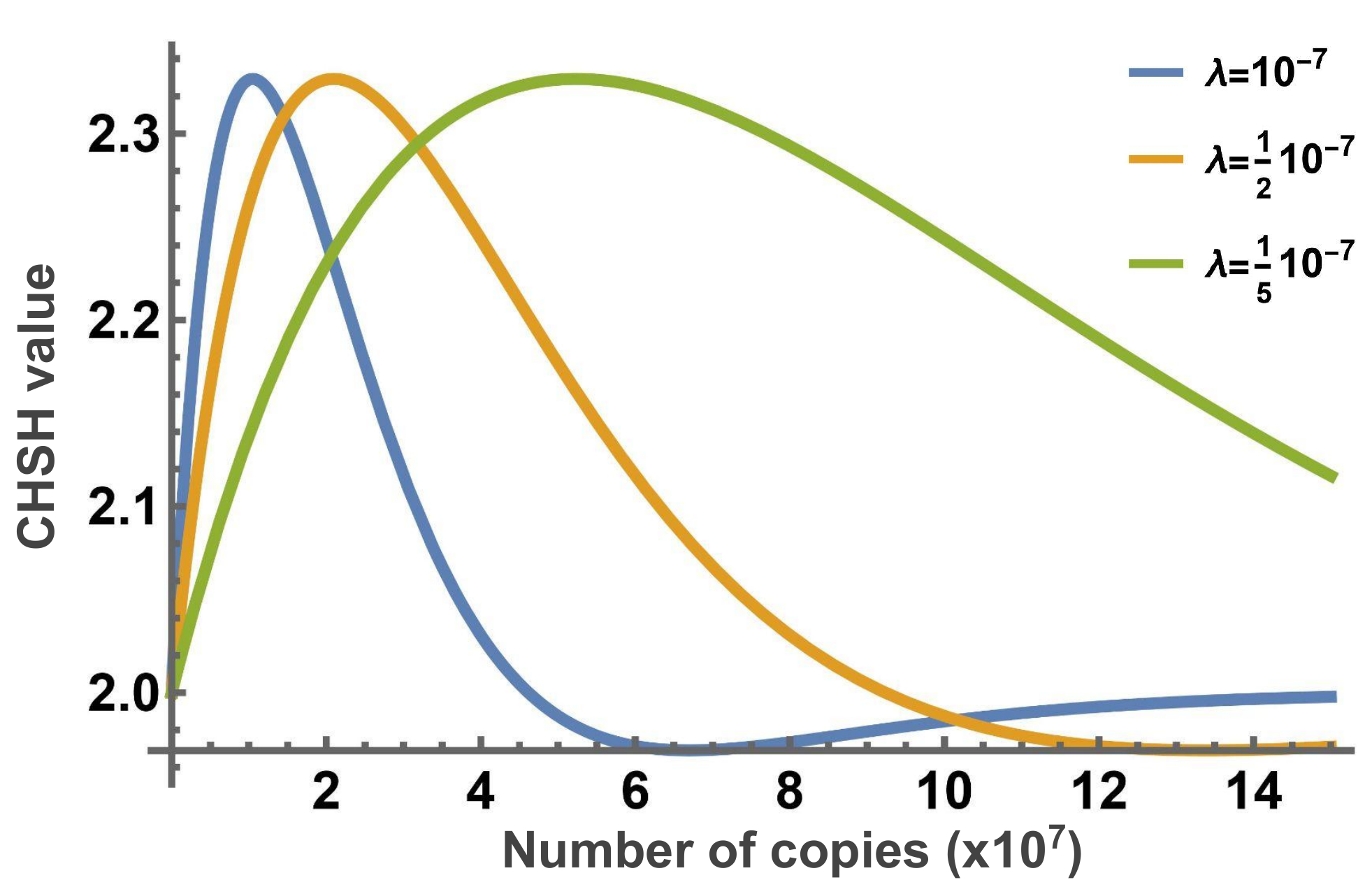}
\caption{The distilled CHSH value of the correlation $\tilde{\mathrm{B}}_Q(\lambda)$ is plotted as a function of the number of copies for different values of $\lambda$. The figure shows that optimal distilled CHSH value is around 2.32928 for the different values of small $\lambda$ considered, yielding the Tsirelson gain, $\Delta\mathcal{T}\approx\frac{0.32928}{2\sqrt{2}-2}\times 100\%\approx39.748\%$}.
\label{Fig:D1}
\end{figure}

Let us start by considering a small value of $\lambda=10^{-7}$ and studying the variation of the distilled CHSH value with the number of copies (see Fig.\ref{Fig:D1}). As seen from Fig.\ref{Fig:D1}, the distilled value for $\lambda=10^{-7}$ initially increases with the number of copies, reaches a maximum of 2.32928 for $1.04739\times10^7$ copies and then decays to the classical value of 2. Note that the CHSH value drops below 2 before eventually saturating at 2. Moreover, Fig.\ref{Fig:D1} shows that peak value does not change significantly as $\lambda$ is reduced to $\lambda=\frac{1}{2}10^{-7}$ or $\frac{1}{5}10^{-7}$ and that overall behaviour is similar upto some re-scaling. 

This suggests that the optimal number of copies is inversely proportional to $\lambda$, in the small $\lambda$ limit. This motivates an ansatz for the optimal number of copies $N^{\text{opt}}=\frac{\alpha}{\lambda}$, for $\alpha$ unknown. Since the derivative for $\mathcal{B}^{(n)}(\lambda)$ with $n$ vanishes at the optimal number of copies, solving $\lim_{\lambda\rightarrow0}\frac{d\mathcal{B}^{(n)}(\lambda)}{dn} \Bigr\rvert_{\frac{\alpha}{\lambda}}=0$ can be used to find $\alpha$. This can in turn be done using Taylor expansion, yielding the value of $\alpha=1.047391$. One can further numerically show that this ansatz is robust in the sense that for $\lambda<10^{-5}$, $\frac{d\mathcal{B}^{(n)}(\lambda)}{dn} \Bigr\rvert_{\frac{\alpha}{\lambda}} < 10^{-11}$. Thus, using this ansatz, we can obtain the following expression for $\mathcal{B}^{(N^{\text{opt}})}(\lambda)$ when $\lambda < 10^{-5}$:
\begin{align}
    \mathcal{B}^{(N_{\text{opt}})}(\lambda)\approx2.32928 + 0.169161\lambda+O(\lambda^2).
\end{align}
This proves the claim that as $\lambda\rightarrow 0$, the AND-OR protocol can distil a large non-locality of approximately 2.32928.

\section{Appendix D: Extended proof of Theorem 3}
Given two correlations $\chi_1\equiv\{p_1(ab|xy\}$ and $\chi_2\equiv\{p_2(ab|xy\}$ in $\mathcal{NS}$, the resulting correlation obtained through OR-AND wiring is denoted as $\mathcal{W}[\chi_1,\chi_2]\equiv\{p(ab|xy\}$. The input-output joint probabilities for the resulting correlation reads as
\begin{subequations}
\begin{align}
p(00|xy)&=\sum_{b_1\wedge b_2=0}p_1(0b_1|xy)p_2(0b_2|xy);\\
p(01|xy)&=p_1(01|xy)p_2(01|xy);\\
p(11|xy)&=\sum_{a_1\vee a_2=1}p_1(a_11|xy)p_2(a_21|xy);\\
p(10|xy)&=1-p(00|xy)-p(01|xy)-p(11|xy);
\end{align}
\end{subequations}
for all $x,y\in\{0,1\}$. Now, let us consider two correlation $\mathrm{C}\equiv\{c(ab|xy)\}\in \mathcal{NS}$ and $\mathrm{P}_{L_1}=\{p_{_{L_1}}(ab\vert xy)=\delta_{a,0}~\delta_{b,1}\}$. We denote $\mathcal{W}[\mathrm{C},\mathrm{P}_{L_1}]\equiv\{p_{_{\mathrm{C},\mathrm{L}_1}}(ab|xy)\}$, then on substituting the probabilities of the local box $\mathrm{P}_{L_1}$, above stated relations for two OR-AND wired boxes reduce to:
\begin{subequations}
\begin{align}
p_{_{\mathrm{C},\mathrm{L}_1}}(00|xy)&=c(00\vert xy);\\
p_{_{\mathrm{C},\mathrm{L}_1}}(01|xy)&=c(01\vert xy);\\
p_{_{\mathrm{C},\mathrm{L}_1}}(10|xy)&=c(10\vert xy);\\
p_{_{\mathrm{C},\mathrm{L}_1}}(11|xy)&=c(11\vert xy);
\end{align}
\end{subequations}
for all $x,y\in\{0,1\}$. Therefore, $\mathcal{W}[\mathrm{C},\mathrm{P}_{L_1}]=\mathrm{C}$. Similarly, by substituting the probabilities of the $\mathrm{P}_{L_1}$ box it also easily follows that $\mathcal{W}[\mathrm{P}_{L_1},\mathrm{C}]=\mathrm{C}$,  and $\mathcal{W}[\mathrm{P}_{L_1},\mathrm{P}_{L_1}]=\mathrm{P}_{L_1}$.

Further, for obtaining a quantitative bound on the the radius of the eight-dimensional Ball centered at $\mathrm{P}_{L_1}$ within which any nonlocal correlation, in a sector of nonzero-measure, allows $2$-copy OR-AND distillation, we write the correlation $\mathrm{C}$ in terms of the coefficients of the PR-box and eight local boxes, then
\begin{equation}
\mathrm{C}= c_0\mathrm{P}_{NL}+\sum_{i=1}^{8}c_i\mathrm{P}_{L_{i}}.
\end{equation}
Now on calculating, the respective Bell-CHSH values $\mathcal{K}$ and $\mathcal{K}^{(2)}$ of respective correlations $\mathrm{C}$ and $\mathrm{C}^{(2)}$ we find that
\begin{align}
\mathcal{K}&= 2(1+ c_0),\\
\mathcal{K}^{(2)}&= 2+c_0^2+4 c_0 (c_1+2 c_4)+8 c_4 (c_1+c_2+c_3+c_5+c_6-1)+8 c_4^2-8 c_5 c_7.
\end{align}
In the above we have used the normalization condition $\sum_{i=0}^8c_i=1$ to eliminate the coefficient $c_8$. Then the condition for two copy distillation can be written as follows: 
\begin{align}
(\mathcal{K}-2)+(\mathcal{K}^{(2)}-2\mathcal{K}+2)~\lambda&>0 \nonumber \\
\implies ~2~c_0+ \left\{c_0^2+4 c_0 (c_1+2 c_4-1)+8 \left(c_4 (c_1+c_2+c_3+c_5+c_6-1)+c_4^2-c_5 c_7\right)\right\}\lambda &>0\nonumber \\
\equiv~2~c_0+ f(c_0,c_1,c_2,c_3,c_4,c_5,c_6,c_7)~\lambda &>0.
\end{align}
Since $2 c_0 + \min\{f(c_0,c_1,...,c_7)\}>0 \implies 2 c_0 + f(c_0,c_1,...,c_7) >0$, two copy distillation is possible if $2 c_0 + \min\{f(c_0,c_1,...,c_7)\}>0$. Further, $\min\{f(c_0,c_1,...,c_7)\}=-3$ and it is achieved at $c_0=1$ and $c_1=c_2=c_3=c_4=c_5=c_6=c_7=0$. Therefore, we obtain that for any $NS$ correlation $\tilde{\mathrm{C}}(\lambda)$, two copy distillation is successful if
\begin{align}
2c_0-3\lambda >0 \Rightarrow \lambda <\frac{2}{3}~c_0.
\end{align}
Thus we also have an analytical expression for the radius of the Ball centered at $\mathrm{P}_{L_1}$ such that any nonlocal no-signaling correlation $\tilde{\mathrm{C}}(\lambda)$, be it quantum or post-quantum, chosen from a nonzero-measure sector of the Ball can be distilled (in full eight dimensions).

\section{Appendix E: Post-quantum certification through OR-AND distillation}
In the main manuscript we have already discussed that the correlation $\mathrm{H}_{NS}$ given by,
\begin{align}\label{corr}
\mathrm{H}_{NS}&=0.1~\mathrm{P}_{NL}+0.85~\mathrm{P}_{L_1}+0.01~\mathrm{P}_{L_2}+0.01~\mathrm{P}_{L_3}+0.02~\mathrm{P}_{L_4}+0.01~\mathrm{P}_{L_5},\\
&=\begin{blockarray}{ccccc}
xy/ab& ~~~~00~~~~ & ~~~~01~~~~ & ~~~~10~~~~ & ~~~~11~~~~ \\
\begin{block}{c(cccc)}
00 & 0.05 & 0.85 & 0.01 & 0.09 \\
01 & 0 & 0.90 & 0.08 & 0.02 \\
10 & 0 & 0.93 & 0.06 & 0.01 \\
11 & 0.01 & 0.92 & 0.07 & 0 \\
\end{block}
\end{blockarray}
\end{align}
is post quantum. While its Hardy success is $0.05$, after OR-AND distillation the Hardy success gets increased. The optimal distillation is achieved for $8$-copy parent correlation and the resulting child correlation reads as
\begin{align}
\mathrm{H}^{(8)}_{NS} &=0.31594\mathrm{P}_{NL}+0.27249~\mathrm{P}_{L_1}+0.23246~\mathrm{P}_{L_2}+0.04999~\mathrm{P}_{L_3}+0.08275~\mathrm{P}_{L_4}+0.0436~\mathrm{P}_{L_5},\\
&=\begin{blockarray}{ccccc}
xy/ab& ~~~~~~~~00~~~~~~~~ & ~~~~~~~~01~~~~~~~~ & ~~~~~~~~10~~~~~~~~ & ~~~~~~~~11~~~~~~~~ \\
\begin{block}{c(cccc)}
00 & 0.157977 & 0.272491 & 0.232454 & 0.337078 \\
01 & 0 & 0.430467 & 0.486781 & 0.0827517 \\
10 & 0 & 0.559582 & 0.390431 & 0.0499871 \\
11 & 0.0463629 & 0.513219 & 0.440418 & 0 \\
\end{block}
\end{blockarray}
\end{align}
Child correlation  $\mathrm{H}^{(8)}_{NS}$ having the Hardy success $0.157977$, which is greater that optimal quantum bound $(5\sqrt{5}-11)/2\approx0.09$, this fact then establishes the  post-quantum nature of the parent correlation $\mathrm{H}_{NS}$.  

Importantly, several other known post-quantum tests fail to detect post-quantumness of the correlation $\mathrm{H}_{NS}$. For instance, correlations having CHSH value more than $4\sqrt{2/3}\approx 3.26599$ violate nontrivial communication complexity \cite{Brassard2005}. However, for the correlations $\mathrm{H}_{NS}$ and $\mathrm{H}^{(8)}_{NS}$ the CHSH values turns out to be $2.2$ and $2.63088$, respectively, which are strictly less than the aforesaid threshold value $3.26599$.

The authors in \cite{Pawlowski2009} have obtained a sufficient criterion for violation of the information causality principle. For the CHSH symmetry considered in this work, the condition reads as
\begin{align}
\mathcal{I}:=E_1^2+E_2^2>1,~~&\mbox{where}~~E_i:=2\mathcal{P}_i-1,~~\mbox{with}\\
\mathcal{P}_1:=\frac{1}{2}[p(a\oplus b=0|00)+p(a\oplus b=1|10)],~~~
&\mathcal{P}_2:=\frac{1}{2}[p(a\oplus b=1|01)+p(a\oplus b=1|11)].
\end{align}
For the correlations $\mathrm{H}_{NS}$ and $\mathrm{H}^{(8)}_{NS}$ the values that $\mathcal{I}$ attain are respectively $0.9578$ and $0.9565$. Therefore information causality remains silent about the post-quantum nature of these boxes. 

However, it may be noted that correlation given in Eq.(\ref{corr}) falls outside the first level of the NPA-hierarchy \cite{Navascues2008}. This easily proves the post-quantumness of the correlation under consideration, and moreover, implies that the correlation violates the principle of macroscopic locality \cite{Navascues2009}.

That being said, we do point out that checking membership to the different levels of NPA hierarchy can become computationally expensive, particularly at higher levels of the hierarchy, while the distillation criteria can be far more computationally tractable. For instance, let us consider the correlation:
\begin{align}\label{corr2}
\mathrm{H}^\prime_{NS}&=\begin{blockarray}{ccccc}
xy/ab& ~~~~00~~~~ & ~~~~01~~~~ & ~~~~10~~~~ & ~~~~11~~~~ \\
\begin{block}{c(cccc)}
00 & 0.0773 & 0.0256 & 0.5599 & 0.3372 \\
01 & 0  & 0.1029 & 0.7804 & 0.1167 \\
10 & 0 & 0.3374 & 0.6372 & 0.0254 \\
11 & 0.1178 & 0.2196 & 0.6626 & 0 \\
\end{block}
\end{blockarray}~~~~.
\end{align}
One may easily verify that correlation given in Eq.(\ref{corr2}) is post-quantum by observing that after 2-copy distillation using OR-AND protocol, the Hardy success goes up to 0.0925 (a value beyond the maximum possible success probaility in quantum mechanics). On the other hand, for the considered correlation, tests like known necessary conditions for violating non-trivial computational complexity and information causality principle fails to detect its post-quantumness. However, unlike the first example, on considering NPA criteria, a membership test into the second tier of the NPA hierarchy is required to establish the post-quantumness of this correlation. Along similar lines, one may imagine post-quantum correlations like Eq.(\ref{corr2}) that lies at further deeper levels of the NPA-hierarchy, while its post-quantumness may be conveniently detected via efficient nonlocality distillation protocols.

\end{document}